\documentclass[10pt,a4paper]{article}
\usepackage[utf8]{inputenc}
\usepackage{enumitem}
\usepackage{authblk}
\usepackage[colorlinks = true,
            linkcolor = blue,
            urlcolor  = blue,
            citecolor = blue,
            anchorcolor = blue]{hyperref}

\usepackage{bbm}
\usepackage{amsmath,amsthm}
\usepackage{alltt, amssymb}
\usepackage[]{algorithm2e}
\usepackage{url}
\usepackage{graphicx}

\newcommand{\bN}{ {\mathbb  N}}
\newcommand{\bZ} { {\mathbb Z}}
\newcommand{\bQ}{ {\mathbb  Q}}
\newcommand{\bC}{ {\mathbb  C}}
\newcommand{\bK}{ {\mathbb  K}}

\newcommand{\si} { {\sigma}}

\newcommand{\lc}{ \operatorname{lc}}

\newcommand{\lcm}{ \operatorname{lcm}}

\newcommand{\pa}{ {\partial}}

\newcommand{\dis}{ \operatorname{dis}}

\newcommand{\qpol}{ {\bK(q)[x][\pa]}}
\newcommand{\qrat}{ {\bK(q, x)[\pa]}}

\newcommand{\cont}{\operatorname{Cont}}

\newtheorem{thm}{Theorem}[section]
\newtheorem{cor}[thm]{Corollary}
\newtheorem{lemma}[thm]{Lemma}
\newtheorem{prop}[thm]{Proposition}
\newtheorem{conj}[thm]{Conjecture}
\newtheorem{defn}[thm]{Definition}
\newtheorem{ex}[thm]{Example}
\newtheorem{algo}[thm]{Algorithm}

\begin{document}

\title{Desingularization in the $q$-Weyl algebra}

\author{Christoph Koutschan\thanks{Supported by the Austrian Science Fund (FWF): P29467-N32.\\
  E-mail addresses: christoph.koutschan@ricam.oeaw.ac.at, zhangy@amss.ac.cn}\ }
\author{Yi Zhang$^{\ast}$}
\affil{Johann Radon Institute for Computational and Applied Mathematics (RICAM), Austrian Academy of Sciences, Austria}



\date{}
\maketitle

\begin{abstract}
In this paper, we study the desingularization problem in the first $q$-Weyl
algebra.  We give an order bound for desingularized operators, and thus derive
an algorithm for computing desingularized operators in the first $q$-Weyl
algebra.  Moreover, an algorithm is presented for computing a generating set
of the first $q$-Weyl closure of a given $q$-difference operator.  As an
application, we certify that several instances of the colored Jones polynomial
are Laurent polynomial sequences by computing the corresponding desingularized
operator.
\end{abstract}

\section{Introduction}\label{SECT:intro}

The desingularization problem has been primarily studied in the context of
differential operators, and more specifically, for linear differential
operators with polynomial coefficients. The solutions of such an operator are
called \emph{D-finite}~\cite{Stanley1980} or \emph{holonomic} functions. 
It is well known~\cite{Ince1926} that a singularity (e.g., a pole) at 
a certain point~$x_0$ of one of the solutions
must be reflected by the vanishing (at~$x_0$) of the leading coefficient of
the operator. The converse however is not necessarily true: not every zero of
the leading coefficient polynomial induces a singularity of at least one
function in the solution space. The goal of desingularization is to construct
another operator, whose solution space contains that of the original operator,
and whose leading coefficient vanishes only at the singularities of the
previous solutions. Typically, such a desingularized operator will have a
higher order, but a lower degree for its leading coefficient.  In summary, desingularization
provides some information about the solutions of a given differential
equation.

For linear ordinary differential and recurrence equations, 
desingularization has been extensively studied 
in~\cite{Abramov1999, Abramov2006, Chen2013, Chen2016, Barkatou2015}. 
Moreover, the authors of~\cite{Yi2017} develop algorithms for the multivariate case. 
As applications, the techniques of desingularization can be used to 
extend P-recursive sequences~\cite{Abramov1999}, 
certify the integrality of a sequence~\cite{Abramov2006}, 
check special cases of a conjecture of Krattenthaler~\cite{Zhang2016} and 
explain order-degree curves~\cite{Chen2013} for Ore operators. 

The authors of~\cite{Chen2016, Zhang2016} also give general algorithms for the Ore case. 
However, from a theoretical point of view, the story is not yet finished, 
in the sense that there is no order bound for desingularized operators in the Ore case. 
In this paper, we consider the desingularization problem in the first $q$-Weyl algebra. 
Our main contribution is to give an order bound (Theorem~\ref{THM:orderbound}) 
for desingularized operators, 
and thus derive an algorithm (Algorithm~\ref{ALGO:desinop}) 
for computing desingularized operators in the first $q$-Weyl algebra.  
In addition, an algorithm (Algorithm~\ref{ALGO:qWeylclosure}) is presented 
for computing a generating set of 
the first $q$-Weyl closure of a given $q$-difference operator.

As an example, consider the $q$-holonomic sequence
\[
  f(n) = [n]_q := \frac{q^n-1}{q-1}
\]
that is a $q$-analog of the natural numbers.  The minimal-order homogeneous
$q$-recurrence satisfied by $f(n)$ is
\[
  (q^n-1)f(n+1) - (q^{n+1}-1)f(n) = 0,
\]
in operator notation:
\begin{equation}\label{EQ:op1}
  \bigl((x-1)\pa - q x + 1\bigr) \cdot f(n) = 0,
\end{equation}
where $x = q^n$ and $\pa \cdot f(n) = f(n + 1)$. 
When we multiply this operator by a suitable left factor, we obtain
a monic (and hence: desingularized) operator of order~$2$:
\begin{equation}\label{EQ:op2}
  \frac{1}{q x-1}\bigl(\pa - q\bigr)\bigl((x-1)\pa - q x + 1\bigr) =
  \pa^2 - (q+1)\pa + q.
\end{equation}

As it is typically done in the shift case~\cite{Abramov1999}, we view a
$q$-difference operator as a tool to define a $q$-holonomic
sequence. Alternatively, one could take the viewpoint of~\cite{Abramov2006}
and study solutions of $q$-recurrences that are meromorphic functions in the
complex plane (for this, let $q\in\bC$ be transcendental), and whose poles are
somehow related to the zeros of the leading coefficient. In that sense, the
factor $x-1$ in~\eqref{EQ:op1} indicates that there may be a pole at $x=q$,
but in fact, the solution $f(x)=\frac{x-1}{q-1}$ is an entire function and
does not have any pole, which is in agreement with the fact that there exists
a desingularized operator~\eqref{EQ:op2}. However, in contrast to the
differential case, in the shift case one also has to take into account poles
that are \emph{congruent}~\cite{Abramov2006} to a zero of the leading
coefficient.  We expect that the same phenomenon occurs in the $q$-case, but
since our main interest is in sequences, we do not investigate it in more
detail here. 

As an application, we study several instances of the colored Jones
polynomial~\cite{GaroufalidisLe05,GaroufalidisSun,GaroufalidisKoutschan12a},
which is a $q$-holonomic sequence arising in knot theory and which is a
powerful knot invariant. By inspecting this sequence for a particular given
knot, one finds that all its entries seem to be Laurent polynomials, and not,
as one would expect, more general rational functions in~$q$. By computing the
corresponding desingularized operator, we can certify that the sequence under
consideration actually is constituted of Laurent polynomials, and that no
other denominators than powers of~$q$ can appear.

\section{Rings of \texorpdfstring{$q$}{q}-difference operators} \label{SECT:ringqd}
Throughout the paper, we assume that $\bK$ is a field of characteristic zero, 
and $q$ is transcendental over $\bK$.
For instance, $\bK$ can be the field of complex numbers and $q$ a transcendental indeterminate.
Let $\bK(q)[x]$ be the ring of usual commutative polynomials over~$\bK(q)$. 
The quotient field of $\bK(q)[x]$ is denoted by $\bK(q, x)$.
Then we have the \emph{ring of $q$-difference operators with rational function 
coefficients} or \emph{$q$-rational algebra}~$\qrat$,
in which addition is done coefficient-wise and multiplication is defined by associativity via the
commutation rule
\[
 \pa f(x) = f(qx) \pa \ \ \text{ for each } f(x) \in \bK(q, x).
\]
The variable $x$ acts on a function $g(x)$ by the usual multiplication, 
and the $q$-difference operator $\pa$ acts on it by 
the \emph{$q$-dilation} with respect to $x$: 
$$\pa(g(x)) = g(qx).$$
This ring is an Ore algebra~\cite{Robertz2014, Salvy1998}. 

Another ring is $\qpol$, which is a subring of~$\qrat$.
We call it the \emph{ring of $q$-difference operators with polynomial coefficients} or
the \emph{q-Weyl algebra}~\cite[Section 2.1]{Stavros2012}.

Given $P \in  \qpol \setminus \{ 0 \}$, we can uniquely write it as
\[
 P = \ell_r \pa^r + \ell_{r-1} \pa^{r-1} + \cdots + \ell_0
\]
with $\ell_0, \ldots, \ell_r \in \bK(q)[x]$ and $\ell_r \neq 0$.
We call~$r$  the \emph{order}, and~$\ell_r$ ~the \emph{leading coefficient} of $P$.
They are denoted by~$\deg_{\pa}(P)$ and~$\lc_{\pa}(P)$, respectively.
We call~$\ell_0$ the \emph{trailing coefficient} of $P$.
Without loss of generality, we assume that $\ell_0 \neq 0$ throughout the paper. 
Otherwise, let $t$ be the minimal index such that $\ell_t \neq 0$. 
Set $\tilde{P} = \pa^{-t} P$. 
Then the trailing coefficient of $\tilde{P}$ is $\pa^{-t}(\ell_t)$, 
which is a nonzero polynomial in $\bK(q)[x]$. 
As a matter of convention, we say that the zero operator in $\qpol$ has order $-1$. 

Let $\sigma\colon \bK(q)[x] \rightarrow \bK(q)[x]$ be a ring automorphism that leaves the
elements of~$\bK(q)$ fixed and $\sigma(x) = q x$.
Assume that~$Q \in \qpol$ is of order~$k$. A repeated use of the commutation rule yields
\begin{equation} \label{EQ:productlc}
\lc_{\pa}(QP) = \lc_{\pa}(Q) \si^{k}(\lc_{\pa}(P)).
\end{equation}
Assume that $S$ is a subset of $\qrat$, then the left ideal generated by~$S$ is denoted by $\qrat  S$. 
For an operator $P \in \qpol$,
we define the \emph{contraction ideal} or \emph{$q$-Weyl closure} of $P$:
\[
  \cont(P) := \qrat P \;\cap\; \qpol.
\]

\section{Dispersion in the \texorpdfstring{$q$}{q}-case} \label{SECT:qdispersion}
In this section, we define the dispersion of two polynomials in $\bK(q)[x]$ and 
present an algorithm for computing it, based on irreducible factorizations over the ring~$\bK[q][x]$. 
The dispersion in the $q$-case will be used in the next section for giving an order bound of a desingularized operator 
(Definition~\ref{DEF:desingularization}). 

\begin{lemma} \label{LEM:finite}
The following claims hold:
\begin{itemize}
 \item [(i)] If $p(x)$ is an irreducible polynomial in $\bK(q)[x]$ of positive degree with $p(0) \neq 0$, 
 so is $p(q^\alpha x)$ for each $\alpha \in \bZ$.
 \item [(ii)] Let $p(x)$ be an irreducible polynomial in $\bK(q)[x]$ of positive degree with $p(0) \neq 0$. Then 
 \[
  \gcd(p(q^\alpha x), p(x)) = 1 \quad\text{for each } \alpha \in \bZ \backslash \{ 0 \}.
 \]
 \item [(iii)] Let $f(x), g(x)$ be two polynomials in $\bK(q)[x]$ with $f(0) \neq 0$. 
 Then the set 
\begin{equation} \label{EQ:dispersionset}
 \{ \alpha \in \bN \mid \deg_x(\gcd(f(q^\alpha x) , g(x))) > 0 \} 
\end{equation}
 is a finite set. 
\end{itemize}
\begin{proof}
(i) It follows from~\cite[Proposition 3]{Man1994}.

(ii) Suppose that there exists $\alpha_0 \in \bZ \backslash \{ 0 \}$ such that 
$$\gcd(p(q^{\alpha_0} x), p(x)) \neq 1.$$
Since $p(x)$ is an irreducible polynomial in $\bK(q)[x]$, we have that $p(x) \mid p(q^{\alpha_0} x)$.
We may write 
\begin{equation} \label{EQ:poly}
 p(x) = c_d x^d + c_{d - 1} x^{d -1} + \cdots + c_0,
\end{equation}
where $c_i \in \bK(q)$, $0 \le i \le d$ with $d > 0$, and $c_0, c_d \neq 0$. Then 
\begin{equation} \label{EQ:shiftedpoly}
 p(q^{\alpha_0} x) = (c_d q^{d \alpha_0 }) x^d + (c_{d - 1} q^{(d -1) \alpha_0 }) x^{d - 1} + \cdots + c_0.
\end{equation}
Since $p(x) \mid p(q^{\alpha_0} x)$, we conclude from~\eqref{EQ:poly} and~\eqref{EQ:shiftedpoly} that
\[
 p(q^{\alpha_0} x) = q^{d \alpha_0} p(x).
\]
Comparing the constant coefficients of both sides in the above equation, it follows that 
\[
 c_0 q^{d \alpha_0} = c_0.
\]
Since $c_0 \neq 0$, we have that $q^{d \alpha_0} = 1$, a contradiction to the fact that $q$ is not a root of unity of $\bK$. 

(iii) Suppose that~\eqref{EQ:dispersionset} is an infinite set. 
Then there exists an irreducible factor $p(x)$ of $f(x)$ such that
\[
 \gcd(p(q^\alpha x), g(x)) \neq 1 \quad\text{for infinitely many $\alpha \in \bN$}.
\]
Since $g(x)$ only has finitely many distinct irreducible factors, it follows from (i) that 
\[
 \gcd(p(q^{\alpha_1} x), p(q^{\alpha_2} x)) \neq 1 \quad\text{for some } \alpha_1 \neq \alpha_2 \in \bN.
\]
Therefore, we have 
\[
 \gcd(p(q^{\alpha_1 - \alpha_2} x), p(x)) \neq 1,
\]
a contradiction to (ii).
\end{proof}

\end{lemma}

Based on the above lemma and~\cite[Definition 1]{Man1994}, we give the following definition. 

\begin{defn} \label{DEF:qdispersion}
Let $f(x), g(x)$ be two polynomials in $\bK(q)[x]$ with $f(0) \neq 0$. 
The dispersion of $f(x)$ and $g(x)$ is given by 
\[
 \dis(f(x), g(x)) := \max\; \{ \alpha \mid \alpha \in \bN, \deg_x(\gcd(f(q^\alpha x) , g(x))) > 0 \} \cup \{0\}.
\]
\end{defn}

We include $0$ in the above definition in order to guarantee that the dispersion is always defined, 
even for constant polynomials.  
The dispersion in the $q$-case is the largest integer $q$-shift such that 
the greatest common divisor of the shifted polynomial and the unshifted one is nontrivial. 
Specifically, assume that $f(x)$ has the following factorization:
\[
f(x) = p_1^{e_1} \cdots p_m^{e_m},
\]
where $p_1,$ \ldots, $p_m \in \bK(q)[x] \setminus \bK(q)$ are irreducible and pairwise coprime. 
It is straightforward to see from Definition~\ref{DEF:qdispersion} that 
\[
  \dis(f(x), g(x)) = \max \{ \dis(p_i, g) \mid 1 \le i \le m \}.
\]
For example, the dispersion 
$$\dis((x + 1) (4 x + q), (q^2 x +1) (q^3 x + q + 1)) = 2,$$
because $\dis(x + 1, q^2 x + 1) = 2$. 

Similar to the shift case, the dispersion in the $q$-case can be computed by a resultant-based algorithm~\cite[Example 1]{Paule1999}. 
We have implemented it in Mathematica, but experiments suggest that it is inefficient in practice. 
 For instance, consider
\begin{align*}
 f(x) & =  5 (q x + 1) (x - 3 q) (x + 2) (x^3 - q x + 1) (2 q x^3 + 5), \\
 g(x) & = f(q^4 x).
\end{align*}
The polynomial $f$ has coefficients in $\bZ[q]$, and has degree $9$ in $x$. 
The dispersion of $f$ and $g$ is $4$. Below is a table for the timings (in seconds) for the computation of dispersion 
of $f$ and $g$ by the resultant-based (\textbf{res}) algorithm and the factorization-based (\textbf{fac}) algorithm, respectively. For this purpose,
the two polynomials were given in fully expanded form.
\begin{center}
\begin{tabular}{ c | c  }
\hline
System  & Mathematica  \\
\hline
res   &  43.6006 \\
fac &  0.011015  \\
\hline
\end{tabular}
\end{center}

Like~\cite{Man1994}, we also give an algorithm based on irreducible factorization over~$\bK[q][x]$. 

\begin{prop} \label{PROP:primitivedispersion}
Let $f(x)$ be a primitive polynomial in $\bK[q][x]$ of positive degree with respect to $x$, and $f(0) \neq 0$. 
Then for each $\alpha \in \bZ$, we have 
\begin{itemize}
 \item [(i)] $f(q^\alpha x) = q^e g(x)$, where $g(x)$ is a primitive polynomial in $\bK[q][x]$ with the same degree as $f(x)$, $g(0) \neq 0$ 
 and $e \in \bN$.  
 \item [(ii)] Let $f(x) = \sum_{i = 0}^d a_i x^i$ and $g(x) = \sum_{i = 0}^d b_i x^i$ be two polynomials 
 such that $f(q^\alpha x) = q^e g(x)$ for some $e \in \bN$. Then 
 \[
  q^{d \alpha} = \frac{a_0 b_d}{b_0 a_d}.
 \]
\end{itemize}
\end{prop}
\begin{proof}
(i) Assume that $f(x) = \sum_{i = 0}^d a_i x^i$ with $a_d, a_0 \neq 0$, $\gcd(a_d, \ldots, a_0) = 1$ in $\bK[q]$. 
Then
\begin{equation} \label{EQ:qshift}
f(q^\alpha x) = ( a_d q^{d \alpha}) x^d + (a_{d - 1} q^{(d - 1) \alpha}) x^{d - 1} + \cdots + a_0. 
\end{equation}
Since $\gcd(a_d, \ldots, a_0) = 1$ in $\bK[q]$, we have that
\[
 \gcd(a_d q^{d \alpha}, a_{d - 1} q^{(d - 1) \alpha}, \ldots, a_0) = q^e \ \ \text{ for some } e \in \bN.
\]
Thus, we can write $f(q^\alpha x) = q^e g(x)$, where $g(x)$ is a primitive polynomial in~$\bK[q][x]$ 
with the same degree as $f(x)$ and $g(0) \neq 0$.
 
(ii) Since $f(q^\alpha x) = q^e g(x)$, it follows from~\eqref{EQ:qshift} that
\begin{align*}
 \frac{a_0}{a_d q^{d \alpha}} & =  \frac{q^e b_0}{q^e b_d} \\
                              & =  \frac{b_0}{b_d}.
\end{align*}
Thus, we conclude that
\[
  q^{d \alpha} = \frac{a_0 b_d}{b_0 a_d}.
 \]
\end{proof}

Given $f(x), g(x) \in \bK(q)[x]$, we may further assume that $f(x), g(x)$ are two polynomials in $\bK[q][x]$ by 
clearing their denominators.
The above proposition gives a method to compute the dispersion of two primitive irreducible polynomials in $\bK[q][x]$. 
Below is the corresponding algorithm.

\begin{algo} \label{ALGO:irrdispersion}
Given two primitive irreducible polynomials $f, g \in \bK[q][x]$ of positive degrees with respect to $x$ and $f(0) \neq 0$. 
Compute $\dis(f, g)$.
\begin{enumerate}
 \item Compute $d_1 = \deg_x(f)$, $d_2 = \deg_x(g)$. If $d_1 \neq d_2$, then return $0$. 
 Otherwise, set $d = d_1$.
 \item Let $f = \sum_{i = 0}^d a_i x^i$ and $g = \sum_{i = 0}^d b_i x^i$. 
 If $\frac{a_0 b_d}{b_0 a_d}$ is not a nonnegative  power of $q^d$, then return $0$. 
 Otherwise, set $\alpha$ to be the natural number such that $q^{d \alpha} = \frac{a_0 b_d}{b_0 a_d}$.
 \item Compute $h = \frac{f(q^{\alpha} x)}{a_d q^{d \alpha}} - \frac{g(x)}{b_d}$. If $h$ is not the zero polynomial, return $0$.
 Otherwise, return $\alpha$.
\end{enumerate}
\end{algo}

The termination of the above algorithm is obvious. 
The correctness follows from Proposition~\ref{PROP:primitivedispersion}.

\begin{ex} \label{EX: irrdispersion}
Consider the following two primitive polynomials in $\bK[q][x]$:
\begin{align*}
  f(x) &= q x^2 - 1, \\
  g(x) &= q^5 x^2 - 1. 
\end{align*}
Using the above algorithm, we find that $\dis(f(x), g(x)) = 2$. 
\end{ex}

Using the irreducible factorization over $\bK[q][x]$, 
we derive the following algorithm to compute the dispersion of two arbitrary polynomials in $\bK[q][x]$:

\begin{algo} \label{ALGO:dispersion}
Given $f(x), g(x) \in \bK[q][x]$ with $f(0) \neq 0$, compute $\dis(f, g)$.
\begin{enumerate}
 \item{} [Initialize] If $\deg_x(f) < 1$ or $\deg_x(g) < 1$ then return $0$. Otherwise, set {\tt dispersion} $= 0$.
 \item{} [Factorization] Compute the set $\{f_i(x) \}$ and $\{g_j(x)\}$ of distinct primitive 
 irreducible factors over $\bK[q]$ of positive degree in $x$ for $f(x)$ and $g(x)$, respectively.
 \item For each pair $(f_i(x), g_j(x))$ of these factors, use Proposition~\ref{PROP:primitivedispersion} 
 to compute $\alpha = \dis(f_i(x), g_j(x))$. If $\alpha >$ {\tt dispersion}, then set {\tt dispersion}~$= \alpha$.
 \item Return {\tt dispersion}.
\end{enumerate}
\end{algo}

The termination of the above algorithm is obvious. The correctness follows from Definition~\ref{DEF:qdispersion} and 
Proposition~\ref{PROP:primitivedispersion}. It is implemented in Mathematica.

\begin{ex} \label{EX:dispersion}
Consider the following two polynomials in $\bK[q][x]$:
\begin{align*}
  f(x) &= (q x-1 ) (q x+1)(q x^2-1), \\
  g(x) &= q^9 x^7 (q^2 x-1) (q^2 x+1)(q^5 x^2-1). 
\end{align*}
They are already in factored form. 
Using the above algorithm, we find that 
$$\dis(f(x), g(x)) = 2.$$    
\end{ex}

\section{Desingularization in the \texorpdfstring{$q$}{q}-Weyl algebra} \label{SECT:qdesing}

We are now going to present algorithms for the $q$-Weyl closure
(Algorithm~\ref{ALGO:qWeylclosure}) and for the desingularization of a
$q$-difference operator (Algorithm~\ref{ALGO:desinop}). These algorithms are
analogs of those in~\cite{Zhang2016} and use Gr\"obner basis computations.
Hence, in practice, they are slower than algorithms based on linear
algebra~\cite{Chen2013,Chen2016} (see also Section~\ref{SECT:application}),
but their advantage is that also the degree with respect to~$q$ can be
taken into account---a feature that will be essential for the examples
presented in the next section.

In this section, we consider the desingularization for the leading coefficient
of a given $q$-difference operator. The trailing coefficient can be handled in
a similar way.  We summarize some terminologies given in~\cite{Chen2013,
  Chen2016, Zhang2016} by specializing the general Ore ring setting to the
$q$-Weyl algebra.
\begin{defn}\label{DEF:premovable}
Let $P \in \qpol$ with positive order, and~$p$ be a divisor of~$\lc_{\pa}(P)$ in~$\bK(q)[x]$.
\begin{itemize}
 \item[(i)] We say that $p$ is \emph{removable} from $P$ at
order $k$ if there exist~$Q \in \qrat$ with order~$k$, and~$w, v \in \bK(q)[x]$
with $\gcd(p, w) = 1$ in~$\bK(q)[x]$ such that
$$QP \in \qpol \quad \text{and} \quad \sigma^{-k}(\lc_{\pa}(QP)) = \frac{w}{vp}\lc_{\pa}(P).$$
We call $Q$ a \emph{$p$-removing operator for $P$ over $\bK(q)[x]$}, and $QP$ the corresponding \emph{$p$-removed operator}.
\item[(ii)] A polynomial $p \in \bK(q)[x]$ is simply called \emph{removable} from $P$ if it is removable at order $k$ for some $k \in \bN$. 
Otherwise, $p$ is called \emph{non-removable} from $P$.
\end{itemize}
\end{defn}

Note that every $p$-removed operator lies in $\cont(P)$.

\begin{ex} \label{EX:chyzak2010}
Consider the following $q$-difference operator~\cite[Example 4.9]{Chyzak2010} of order $1$ in $\qpol$:
\[
 P = q^2 x (q^2 - x) \pa - (1 - x) (1 - q x).
\]
Set 
\[
 Q = \frac{q^6}{x-1} \pa^2 + \frac{q^6+q^5-q^3-q^2}{x-1} \pa + \frac{q^5-q^3-q^2+1}{x-1}.
\]
Let $L = QP$. Then 
\[
\begin{array}{l@{\;}l@{\;}l}
 L & = & q^{12} x \pa^3 +  q^6(q^5 x + q^4 x + q^3 x - q x - x - 1) \pa^2 +{}\\[4pt]
   &   & (q-1) q^2 (q+1) (q^2+q+1) (q^3 x + q x - x-1) \pa +{} \\[4pt]
   &   & (q-1)^2 (q+1) (q^2+q+1) (qx-1),
\end{array}
\]
is a $(q^2 - x)$-removed operator for $P$ of order $3$.
\end{ex}

The following proposition provides a convenient form of $p$-removing operators over $\bK(q)[x]$. 
It is a special case of~\cite[Lemma 2.4]{Zhang2016} and also included in~\cite{Chen2013}. 
In Corollary~\ref{COR:xnonremovable}, we will use it to prove that $x$-removing operators do not exist.

\begin{prop}\label{PROP:premovable}
Let~$P \in \qpol$ be a $q$-difference operator with positive order. 
Assume that~$p \in \bK(q)[x]$ is removable from~$P$ at order~$k$.
Then there exists a~$p$-removing operator for~$P$ over~$\bK(q)[x]$ of the form
\[
  \frac{p_0}{\sigma^{k}(p)^{d_0}} + \frac{p_1}{\sigma^{k}(p)^{d_1}} \pa + \cdots
  + \frac{p_{k}}{\sigma^{k}(p)^{d_{k}}} \pa^{k},
\]
where $p_i$ belongs to $\bK(q)[x]$, $\gcd(p_i, \sigma^{k}(p)) = 1$ in $\bK(q)[x]$ or $p_i = 0$ for each $i = 0,$ $1,$ \ldots, $k$, and $d_k \geq 1$.
\end{prop}

In~\cite[Lemma 4]{Chen2013}, the authors give an order bound for a $p$-removing operator in the shift case. 
We find that the proof also applies to the $q$-difference case 
provided that $p$ is an irreducible polynomial in $\bK(q)[x]$ and $p(0) \neq 0$. 
We summarize it in the following lemma.

\begin{lemma} \label{LEM:orderbound}
Let $P$ be a nonzero operator in $\qpol$ of positive order with trailing coefficient~$\ell_0$. 
Assume that $p$ is an irreducible factor of $\lc_{\pa}(P)$ such that $p(0) \neq 0$ and $p^k$ is 
removable from $P$ for some $k \ge 1$. 
Then $p^k$ is removable from $P$ at order $\dis(p, \ell_0)$.
\end{lemma}
\begin{proof}
It is literally the same as~\cite[Lemma 4]{Chen2013}.
\end{proof}

Let $P = \sum_{i = 0}^r \ell_i \pa^i$ be a nonzero operator in $\qpol$ of positive order. 
We say that $P$ is $x$-primitive if $x \nmid \gcd(\ell_0, \ldots, \ell_r)$ in $\bK(q)[x]$. 
Gau\ss{'} lemma in the commutative case  also holds for $x$-primitive operators. 
The proof is similar to that of~\cite[Lemma 3.4.8]{Zhang2017}. 
Here, we give an independent proof. 

\begin{lemma} \label{LEM:Gauss}
Let $P$ and $Q$ be two operators in $\qpol$. 
If $P$ and $Q$ are $x$-primitive, so is $QP$.
\end{lemma}
\begin{proof} 
Suppose that $QP$ is not $x$-primitive. 
We may write
\[
P =  \sum_{i = 0}^r a_i \pa^i, \ \ Q = \sum_{i = 0}^s b_i \pa^i \ \ \text{ and } \ \ QP = \sum_{i = 0}^{r + s} c_i \pa^i,
\]
where all coefficients $a_i,b_i,c_i$ are polynomials in $\bK(q)[x]$. 
By assumption, we have $x \mid \gcd(c_0, \ldots, c_{r + s})$. 
Since $P$ and $Q$ are $x$-primitive, 
there exists $0 \le i_0 \le r$ and $0 \le j_0 \le s$ such that~$x \nmid a_{i_0}$ and $x \nmid b_{j_0}$. 
We may further assume that $i_0$ and $j_0$ are maximal with this property. Consider
\begin{equation} \label{EQ:product}
 c_{i_0 + j_0} = \sum_{i + j = i_0 + j_0} a_i \si^{i}(b_j),
\end{equation}
By the maximality of $i_0$ and $j_0$, we have that $x \mid a_i$ and $x \mid b_j$ for $i > i_0$ and~$j > j_0$. 
Note that $x$ also divides $\si^{i}(b_j)$ for $j > j_0$ and $i = i_0 + j_0 - j$ because $\si^{i}(x) = q^i x$. 
Therefore, in the right side of equation~\eqref{EQ:product}, each summand is divisible by $x$ except $a_{i_0} \si^{i}(b_{j_0})$. 
By assumption, $x$ divides $c_{i_0 + j_0}$. Thus, $x$ divides $a_{i_0} \si^{i_0}(b_{j_0})$.
It implies that $x \mid a_{i_0}$ or $x \mid \si^{i_0}(b_{j_0})$. 
Since $x \nmid a_{i_0}$, we have that~$x \mid \si^{i_0}(b_{j_0})$. 
If follows that $x \mid \si^{-i_0} (\si^{i_0}(b_{j_0})) = b_{j_0}$, a contradiction.
\end{proof}

\begin{cor} \label{COR:xnonremovable}
Let $P$ be a nonzero operator in $\qpol$ of positive order. 
If~$x$ divides $\lc_{\pa}(P)$, then $x$ is non-removable from $P$.  
\end{cor}
\begin{proof}
Suppose that~$x$ is removable from $P$. 
By Definition~\ref{DEF:premovable}, there exists an $x$-removing operator~$Q$ such that
$QP \in \qpol$.
By Proposition~\ref{PROP:premovable}, we can write
\[
 Q = \frac{p_0}{x^{d_0}} + \frac{p_1}{x^{d_1}} \pa + \cdots +
     \frac{p_{k}}{x^{d_{k}}} \pa^{k},
\]
where~$p_i \in \bK(q)[x]$, $\gcd(p_i, x) = 1$ in $\bK(q)[x]$, $i = 0, \ldots, k$ and $d_k \geq 1$. Let
$$d = \max_{0 \leq i \leq k} d_i \quad \text{ and } \quad Q_1 = x^d Q.$$
Then the content $w$ of $Q_1$  with respect to $\pa$ is $\gcd(p_0, \ldots, p_k)$
because 
$$\gcd(p_i, x) = 1 \ \ \text{ for each } \ \ i = 0, \ldots, k.$$
Let~$Q_1 = w Q_2.$ Then~$Q_2$ is the primitive part of~$Q_1$. In particular,~$Q_2$ is $x$-primitive.
Then
$$ w Q_2 P = x^d QP.$$
Since $\gcd(w, x) = 1$ and $QP \in \qpol$, we have that $x$ divides the content of~$Q_2P$ with respect to $\pa$.
It follows that $Q_2P$ is not $x$-primitive, a contradiction to Lemma~\ref{LEM:Gauss}.
\end{proof}

Next, we give the definition of desingularized operators in the $q$-case, which is a special case of~\cite[Definition 3.1]{Zhang2016}.

\begin{defn}\label{DEF:desingularization}
Let $P \in \qpol$ with order $r > 0$, and
\begin{equation} \label{EQ:factor}
\lc_{\pa}(P) = p_1^{e_1} \cdots p_m^{e_m},
\end{equation}
where $p_1,$ \ldots, $p_m \in \bK(q)[x] \setminus \bK(q)$ are irreducible and pairwise coprime.
An operator $L \in  \qpol \setminus \{ 0 \} $ of order $k$ is called a {\em desingularized operator for~$P$}
if $L \in \cont(P)$ and
\begin{equation} \label{EQ:dop}
\si^{r - k}(\lc_{\pa}(L)) = \frac{a}{b p_1^{k_1} \cdots p_m^{k_m}} \lc_{\pa}(P) ,
\end{equation}
 where $a, b \in \bK(q)$ with $b \neq 0$, and $p_i^{d_i}$ is non-removable from $P$ for each $d_i > k_i$, $i = 1,\ldots,m$.
\end{defn}

\begin{thm} \label{THM:orderbound}
Let $P$ be a nonzero operator in $\qpol$ of order $r > 0$. 
Assume that $\ell_r$ and $\ell_0$ are the leading and trailing coefficient of $P$, respectively. 
Set $\ell_r = x^e \tilde{\ell_r}$ for some $e \in \bN$ and $\tilde{\ell_r}(0) \neq 0$. 
Then there exists a desingularized operator of $P$ of order $r + \dis(\tilde{\ell_r}, \ell_0)$.
\end{thm}
\begin{proof}
Assume that $\tilde{\ell_r} = p_1^{e_1} \cdots p_m^{e_m}$, 
where $p_1,$ \ldots, $p_m \in \bK(q)[x] \setminus \bK(q)$ 
are irreducible, pairwise coprime. 
For each $i \in \{1, \ldots, m \}$, let $k_i$ be the natural number 
such that $p_i^{k_i}$ is removable from $P$, but $p_i^{d_i}$ is non-removable from $P$ for each $d_i > k_i$. 
It follows from Lemma~\ref{LEM:orderbound} that $p_i^{k_i}$ is removable from $P$ at order $\dis(p_i, \ell_0)$. 
On the other hand, if $e \ge 1$, then it follows from Corollary~\ref{COR:xnonremovable} that~$x^d$ is non-removable from $P$ 
for each $1 \le d \le e$.
Above all, we conclude from~\cite[Lemma 4]{Chen2016} that there exists a desingularized operator of $P$ of order
\[
 r + \max \{ \dis(p_i, \ell_0) \mid 1 \le i \le m \}, 
\]
which is equal to $r + \dis(\tilde{\ell_r}, \ell_0)$.
\end{proof}

\begin{ex} \label{EX:orderbound}
Consider the $q$-difference operator from Example~\ref{EX:chyzak2010}:
\[
 P = q^2 x (q^2 - x) \pa - (1 - x) (1 - q x).
\]
By the above theorem, we find that $P$ has a desingularized operator of order
\[
 1 + \dis(q^2 (q^2 - x), (1 - x) (1 - q x)) = 4.
\]
Actually, a desingularized operator of $P$ with minimal order is $L$ as specified in Example~\ref{EX:chyzak2010}, 
which is of order 3.
\end{ex}

In the above example, the order bound given by Theorem~\ref{THM:orderbound} is overshooting. 
However, we will see in the next section that it is tight in all examples from knot theory
that we looked at.

The first application of Theorem~\ref{THM:orderbound} is to derive an algorithm for 
computing the first $q$-Weyl closure of a $q$-difference operator. 

Let $P$ be a nonzero operator in $\qpol$ of order $r > 0$. 
For each $k \ge r$, we set 
\[
  M_k(P) = \{T \in \cont(P) \mid \deg_{\pa}(T) \le k \}.
\]
It is straightforward 
to see that $M_k(P)$ is a finitely generated left $\bK(q)[x]$-submodule of $\cont(P)$. 
We call it the \emph{$k$-th submodule} of $\cont(P)$. 
If the operator $P$ is clear from the context, then we denote $M_k(P)$ simply by $M_k$. 
A generating set of $M_k$ can be derived by a syzygy computation over $\bK(q)[x]$~\cite[Section 3.3.2]{Zhang2017}. 

\begin{algo} \label{ALGO:qWeylclosure}
Given a $q$-difference operator $P \in \qpol$ of positive order. 
Compute a generating set of the $q$-Weyl closure of $P$.
\begin{enumerate}
 \item Derive an order bound $k$ for a desingularized operator of $P$ by using Theorem~\ref{THM:orderbound}.
 \item Compute a generating set $S$ of $M_k$ by using Gr\"obner bases~\cite[Section 3.3.2]{Zhang2017}.
 \item Return $S$.
\end{enumerate}
\end{algo}

The termination of the above algorithm is obvious. 
The correctness follows from~\cite[Theorem 3.2.3, Corollary 3.2.4]{Zhang2017}.

\begin{ex} \label{EX:qWeylclosure}
Consider the $q$-difference operator in Example~\ref{EX:chyzak2010}:
\[
 P = q^2 x (q^2 - x) \pa - (1 - x) (1 - q x).
\]
From Example~\ref{EX:orderbound}, we know that an order bound for a desingularized operator of $P$ is $4$. 
Using Gr\"obner bases, we can find a generating set of $M_4$. 
Since the size for the generating set of $M_4$ is large, we do not display it here. 
Instead, it follows from Example~\ref{EX:orderbound} that $P$ has a desingularized operator with order~3. 
By~\cite[Theorem 3.2.3, Corollary 3.2.4]{Zhang2017}, the $q$-Weyl closure of $P$ is also generated by~$M_3$. 
Through computation, we find that  $M_3$ is generated by $\{P, L\}$, 
where~$L$ is specified in Example~\ref{EX:chyzak2010}.
\end{ex}

The second application of Theorem~\ref{THM:orderbound} is to give an algorithm 
for computing a desingularized operator of a given $q$-difference operator. 

Let $P$ be a nonzero operator in $\qpol$ of order $r > 0$. 
For each $k \ge r$, let
\[
   I_k = \left\{ [\pa^k] T \mid T \in M_k(P) \right\} , 
\]
where~$[\pa^k] T$ denotes the coefficient of $\pa^k$ in $T$. 
It is straightforward to see that~$I_k$ is an ideal of $\bK(q)[x]$.
We call $I_k$ the \emph{$k$-th coefficient ideal} of $\cont(P)$. 
By~\cite[Lemma 3.3.3]{Zhang2017}, we can compute a generating set of $I_k$ 
if a generating set of $M_k$ is given. 

Assume that $k$ is an order bound for desingularized operators of $P$.
From~\cite[Theorem 3.3.6]{Zhang2017}, an element in $ I_k \setminus \{ 0 \}$ with minimal degree in $x$ will give rise to 
a desingularized operator of $P$. 
In~\cite[Remark 3.3.7]{Zhang2017}, 
the author describes how to use the Euclidean algorithm over $\bK(q)[x]$ to find an element $s$ in $ I_k \setminus \{ 0 \}$  with minimal degree in $x$. 
However, this will in general introduce a polynomial in~$\bK[q]$ when we clear the denominators in $s$. 
In the next section, we will need to find desingularized operators of some $q$-difference operators from knot theory, 
whose leading coefficient is of the form $q^a x^b$, where $a, b \in \bN$. Thus, we shall also minimize the degree of $q$ 
among leading coefficients of desingularized operators of a given $q$-difference operator.  
 Assume that $B \subset \bK[q][x]$ is a generating set of $I_k$.
Next, we give a method that finds an element in $\bK[q][x]$ of $I_k \setminus \{ 0 \}$ with minimal degree in $x$, 
which also has minimal degree in $q$ among nonzero elements of $\langle B \rangle$ in~$\bK[q][x]$ 
with minimal degree in $x$. 

\begin{prop} \label{PROP:gbmin}
Let $B \subset \bK[q][x]$ be a generating set of $I_k$. 
Assume that $G$ is a reduced Gr\"{o}bner basis of the ideal generated by $B$ over $\bK[q][x]$ with respect to the lexicographic order~$q \prec x$. 
Set $g$ to be the element in $G$ with minimal degree in $x$. Then $g$ is also an element in $I_k$ with minimal degree in $x$. 
\end{prop}
\begin{proof}
Assume that $f \in \bK(q)[x]$ is an element in $I_k$ with minimal degree in $x$. 
Since $B = \{b_1, \ldots, b_\ell \}$ is a generating set of $I_k$, we have 
\[
 f = c_1 b_1 + \ldots + c_\ell b_\ell,
\]
where $c_1, \ldots, c_\ell \in \bK(q)[x]$. 
By clearing denominators in the above equation, it follows that
\[
 \tilde{f} = \tilde{c}_1 b_1 + \ldots + \tilde{c}_\ell b_\ell,
\]
where $\tilde{f} = c f$ and $c, \tilde{c}_i \in \bK[q]$, $i = 1, \ldots, \ell$. 
Since $G$ is Gr\"{o}bner basis of the ideal generated by $B$ over $\bK[q][x]$,  the head term
of $\tilde{f}$ is divisible by~$g_i$ for some $g_i \in G$. 
By the choice of the term order, it is straightforward to see that $\deg_x(g_i) \le \deg_x(\tilde{f})$.
On the other hand, the degree of $f$ in $x$ is equal to that of~$\tilde{f}$. 
Thus, $\deg_x(g_i) \le \deg_x(f)$. 
Since $g$ is the element in $G$ with minimal degree in $x$, we have 
\[
 \deg_x(g) \le \deg_x(g_i) \le \deg_x(f).
\]
\end{proof}

\begin{algo} \label{ALGO:desinop}
Given a $q$-difference operator $P \in \qpol$ of positive order. 
Compute a desingularized operator of $P$.
\begin{enumerate}
 \item Derive an order bound $k$ for a desingularized operator of $P$ by using Theorem~\ref{THM:orderbound}.
 \item Compute a generating set $S$ of $M_k$ by using Gr\"obner bases~\cite[Section 3.3.2]{Zhang2017}.
 \item Compute a generating set in $\bK[q][x]$ of $I_k$ by using~\cite[Lemma 3.3.3]{Zhang2017}.
 \item Compute an element $g \in \bK[q][x]$ of $I_k$ with minimal degree in $x$ by using Proposition~\ref{PROP:gbmin}.
 \item Tracing back to the computation of steps (3) and (4), 
 one can find a~$q$-difference operator $L \in \bK[q][x, \pa]$ of $\cont(P)$ 
 such that $\lc_{\pa}(L) = g$. Output~$L$.
\end{enumerate}
\end{algo}

The termination of the above algorithm is evident. 
The correctness follows from~\cite[Theorem 3.3.6]{Zhang2017}. 

\begin{ex} \label{EX:desinop}
Consider the $q$-difference operator in Example~\ref{EX:chyzak2010}:
\[
 P = q^2 x (q^2 - x) \pa - (1 - x) (1 - q x).
\]
\begin{enumerate}
\item By Example~\ref{EX:orderbound}, we know that the minimal order for a
  desingularized operator of $P$ is~ $3$.
\item Using Gr\"obner bases, we can find a generating set of~$M_3$. Since the
  size for the generating set of $M_3$ is large, we do not display it here.
\item By~\cite[Lemma 3.3.3]{Zhang2017}, we find that $I_3$
\footnote{ By computation, we also find that $I_4 = \langle q^{18} x \rangle$. 
This is not a contradiction because $I_4 = \sigma(I_3)$ in $\bK(q)[x]$. 
}
is generated by~$q^{12}x$.
\item It is straightforward to see that $q^{12} x$ is the element in $I_3$
  with minimal degree in $x$.
\item Tracing back to the computation of steps (3) and (4), we find a
  $q$-difference operator $L \in \bK[q][x, \pa]$ of $\cont(P)$, which is exactly the operator 
  in Example~\ref{EX:chyzak2010}.
\end{enumerate}
\end{ex}

\section{Application to knot theory} \label{SECT:application}

In the past years, $q$-difference equations arose naturally in quantum
topology and knot theory.  During the quest for better and better knot
invariants---the ideal invariant would allow to distinguish all knots---the
so-called \emph{colored Jones polynomial} was discovered. The name
\emph{polynomial} is somewhat misleading, as this invariant consists actually
of an infinite sequence of rational functions in $\bQ(q)$ or Laurent
polynomials in $\bQ[q,q^{-1}]$. For the precise definition of the colored
Jones polynomial we refer to~\cite{GaroufalidisLe05}, where it is proven that
for each knot this infinite sequence satisfies a linear $q$-difference
equation with polynomial coefficients, i.e., that the colored Jones polynomial
is always a $q$-holonomic sequence. The same author formulated the following
conjecture.
\begin{conj}[\cite{Garoufalidis18}]\label{conj.Stavros}
Let $J_{K}(n)\in\bQ(q)$ denote the Jones polynomial of a knot~$K$, colored by
the $n$-dimensional irreducible representation of $\mathfrak{sl}_2$ and
normalized by $J_{\mathrm{Unknot}}(n)=1$. Then for the colored Jones polynomial,
i.e., for the sequence $\bigl(J_K(n)\bigr)_{n\in\bN}$ the following holds:
\begin{enumerate}
\item $(1-q^n)J_{K}(n)$ satisfies a bimonic recurrence relation,
\item $J_{K}(n)$ does not satisfy a monic recurrence relation.
\end{enumerate}
\end{conj}

Here, the notion \emph{bimonic} refers to the property that both the leading
and the trailing coefficient are monic (in the sense of
Corollary~\ref{COR:xnonremovable}, i.e., of the form $q^{an+b}$).  Using
desingularization, we can construct such bimonic recurrences, thereby
confirming part (1) of the conjecture in some particular instances.  This
shows that the colored Jones polynomial is actually a sequence of Laurent
polynomials, even when the sequence is extended to the negative integers, by
applying the recurrence into the other direction.  The knot-theoretic
interpretation of this phenomenon is that the substitution $q\to q^{-1}$
corresponds to reversing the orientation of the knot.

\begin{figure}
  \begin{center}
    \includegraphics[height=120pt]{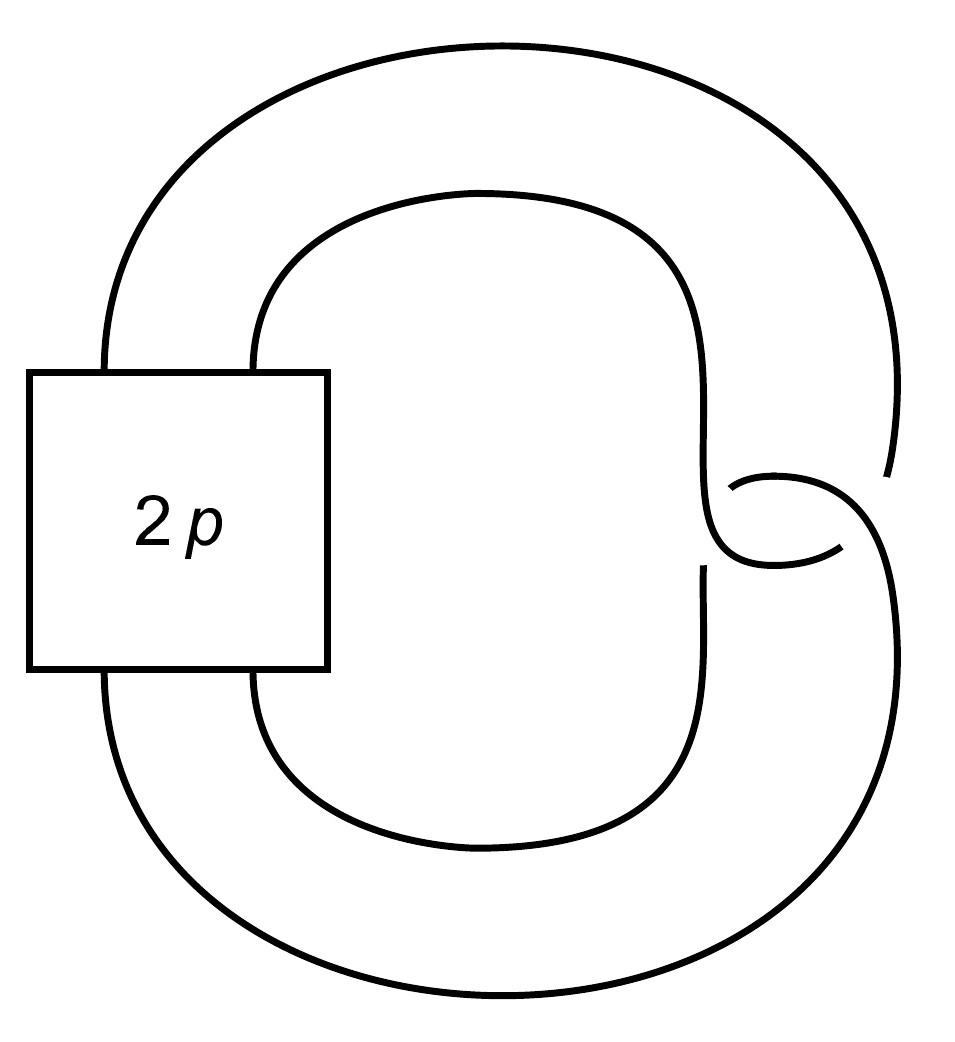}
    \qquad\qquad
    \includegraphics[height=120pt]{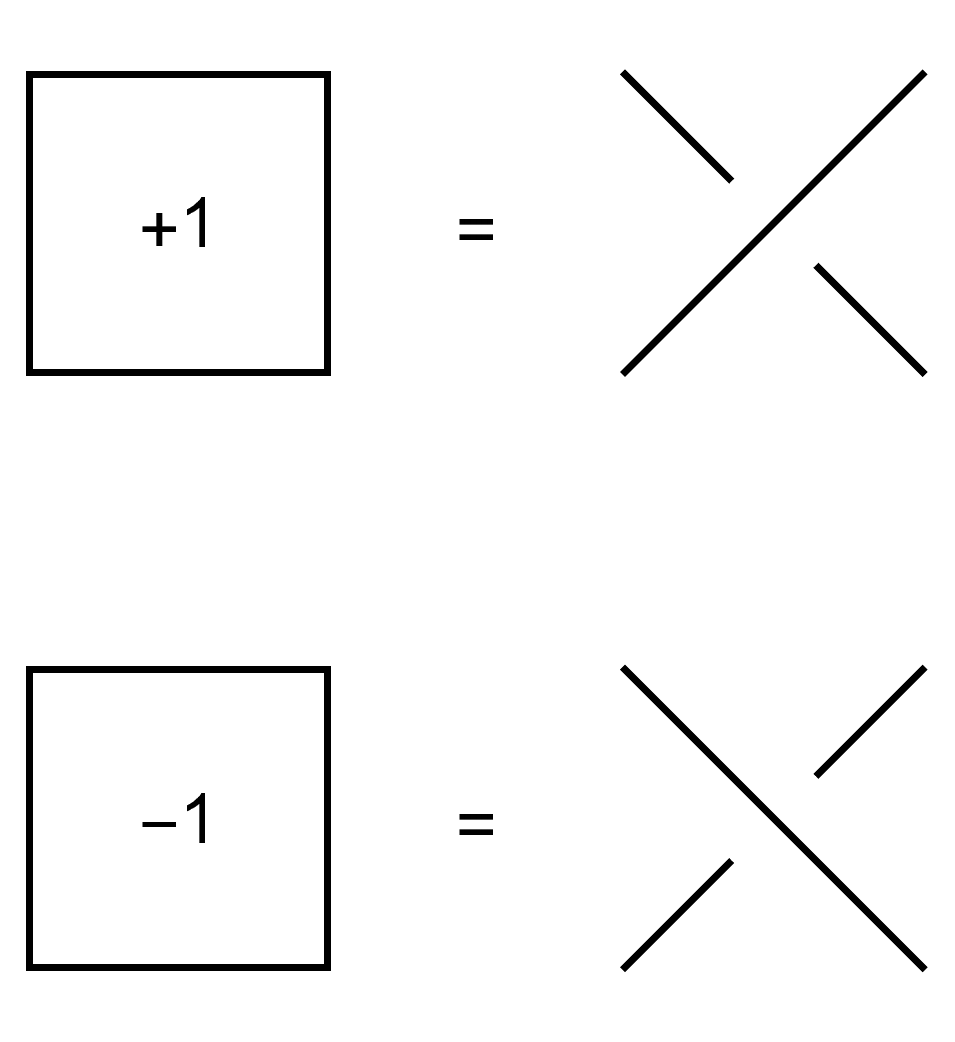}
  \end{center}
  \caption{Knot diagram of the twist knot $K^{\mathrm{twist}}_p$ (left), where
    the box represents repeated half-twists, according to the legend on the
    right.}
  \label{FIG:twist}
\end{figure}
\begin{figure}
  \begin{center}
    \includegraphics[height=120pt]{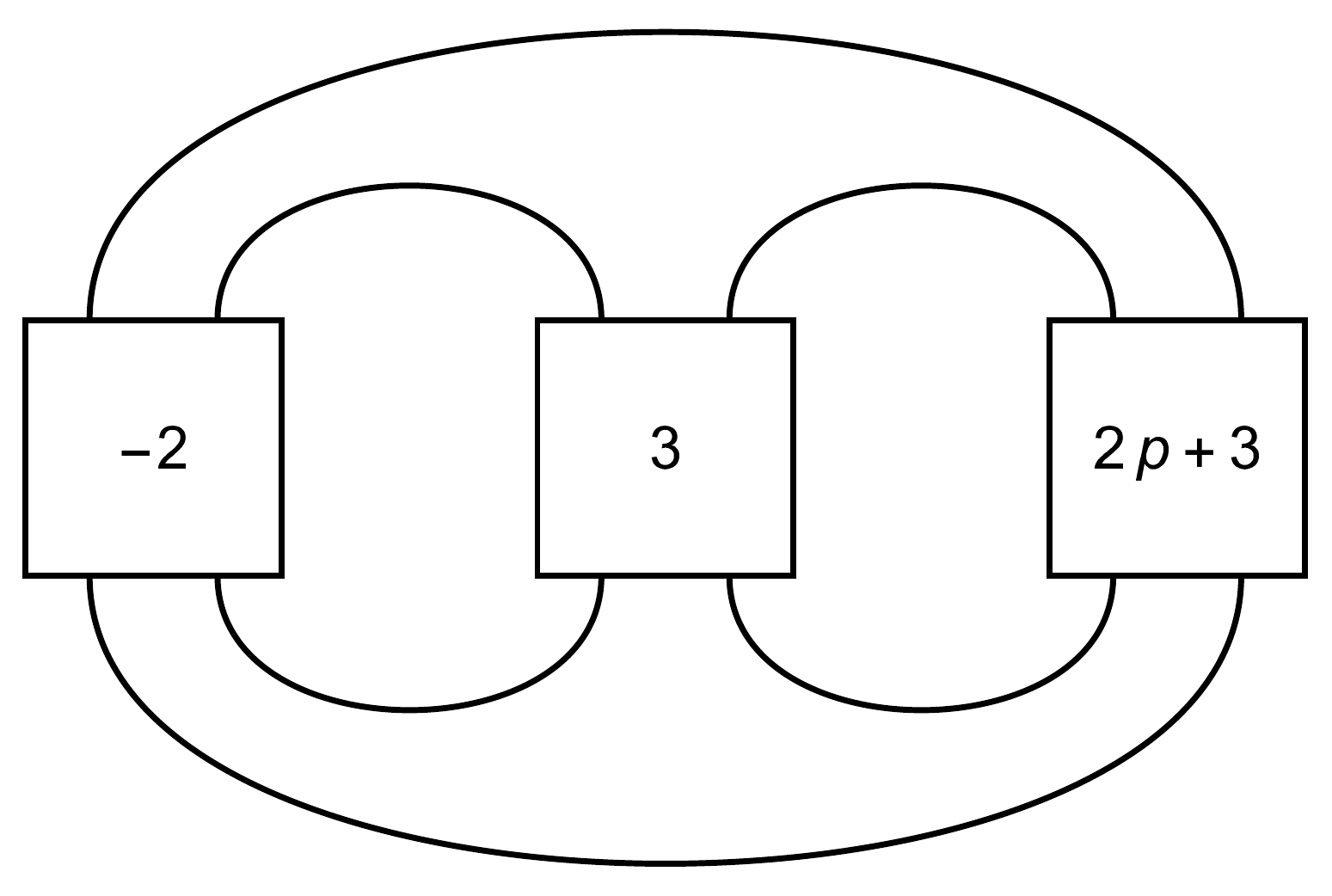}
  \end{center}
  \caption{Knot diagram of the $(-2,3,2p+3)$-pretzel knot
    $K^{\mathrm{pretz}}_p$; again the boxes represent repeated half-twists as
    described in Fig.~\ref{FIG:twist}.}
  \label{FIG:pretz}
\end{figure}

We investigate the colored Jones polynomials of two families of knots that
appeared previously in the literature: twist knots~\cite{GaroufalidisSun} and
pretzel knots~\cite{GaroufalidisKoutschan12a}, see Figures~\ref{FIG:twist}
and~\ref{FIG:pretz}.  While it is very difficult to compute the colored Jones
polynomial for an arbitrary given knot, one can give simpler formulas for
these two families.  For example, the $n$-th entry $J^{\mathrm{twist}}_p(n)$
of the colored Jones polynomial for the $p$-th twist knot
$K^{\mathrm{twist}}_p$ is given by the double sum
\[
  \sum_{k=0}^n \sum_{j=0}^k (-1)^{j+1} q^{k+pj(j+1)+j(j-1)/2} \bigl(q^{2 j+1}-1\bigr)
  \frac{\bigl(q^{1-n};q\bigr)_{\!k} \, \bigl(q^{1+n};q\bigr)_{\!k} \,
  \bigl(q^{k-j+1};q\bigr)_{\!j}}{\bigl(q;q\bigr)_{\!k+j+1}}.
\]
From this representation it is a routine task (but possibly computationally
expensive) to compute a $q$-holonomic recurrence equation for
$J^{\mathrm{twist}}_p(n)$ when $p$ is a fixed integer. This can be done either
by $q$-holonomic summation methods (as implemented in the \texttt{qMultiSum}
package~\cite{Riese03} or \texttt{HolonomicFunctions}
package~\cite{Christoph2010}) or by guessing (as implemented in the
\texttt{Guess} package~\cite{Kauers2009a}). For example, for $p=-1$ we obtain
the inhomogeneous $q$-recurrence
\begin{align*}
  & q^{2 n+2} \left(q^{n+2}-1\right) \left(q^{2 n+1}-1\right) J^{\mathrm{twist}}_{-1}(n+2)
    \;+\; \left(q^{n+1}-1\right)^2 \left(q^{n+1}+1\right) \bigl(q^{n+1}+{} \\
  & +q^{2 n+1}+q^{2 n+3}+q^{3 n+3}-q^{4 n+4}-1\bigr) J^{\mathrm{twist}}_{-1}(n+1)
    \;+\; q^{2n+2} \left(q^n-1\right) \\
  & \times\left(q^{2 n+3}-1\right) J^{\mathrm{twist}}_{-1}(n)
    \;=\; q^{n+1} \left(q^{n+1}+1\right) \left(q^{2 n+1}-1\right) \left(q^{2 n+3}-1\right).
\end{align*}

Garoufalidis and Sun have computed such an inhomogeneous $q$-recurrence
equation for each twist knot $K^{\mathrm{twist}}_p$ with $-15\leq p\leq15$;
the recurrences are available in electronic form
from~\cite{GaroufalidisSun}. Similarly, the $q$-recurrences satisfied by
$J^{\mathrm{pretz}}_p(n)$ for $-5\leq p\leq5$ are available
from~\cite{GaroufalidisKoutschan12a}. By observing that in each recurrence the
term $f(n+d)$ has (among others) a factor $(q^{n+d}-1)$, it is reasonable to
perform the substitution $f(n)\to f(n)/(q^n-1)$,  according to
Conjecture~\ref{conj.Stavros}. In the rest of this section, we only use
operators that were normalized in this way.

We have implemented Algorithms~\ref{ALGO:qWeylclosure} and~\ref{ALGO:desinop}
in Mathematica by using the packages
\texttt{HolonomicFunctions}~\cite{Christoph2010} and
\texttt{Singular}~\cite{Kauers2007}; the source
code and a demo notebook are freely available as part of the supplementary
electronic material~\cite{ElectronicYZ}. 
Note that we also modify Algorithm~\ref{ALGO:desinop} for desingularization 
of the trailing coefficient of a given $q$-difference operator 
in the corresponding package and notebook. 
We give an example about finding
desingularized operators in the context of knot theory.

\begin{ex} \label{EX:knots}
We consider the $q$-difference operators that correspond to the homogeneous
parts of the recurrences for the colored Jones polynomials of the knots
$K^{\mathrm{twist}}_{-1}$, $K^{\mathrm{twist}}_{2}$,
$K^{\mathrm{pretz}}_{-2}$, and $K^{\mathrm{pretz}}_{2}$.  For example, the
operator $P^{\mathrm{twist}}_{-1}$ corresponds, after normalization, to the
left-hand side of the above $q$-recurrence for $J^{\mathrm{twist}}_{-1}(n)$:
\begin{align*}
  P^{\mathrm{twist}}_{-1} ={}& q^2 x^2 \bigl(q x^2-1\bigr) \pa^2 - {} \\
  & (q x-1) (q x+1) \bigl(q^4 x^4-q^3 x^3-q^3 x^2-q x^2-q x+1\bigr) \pa +{} \\
  & q^2 x^2 \bigl(q^3 x^2-1\bigr)
\end{align*}
For space reasons, the other three operators are displayed in abbreviated form
only:
\begin{align*}
  P^{\mathrm{twist}}_{2} &= (qx-1) (qx+1) (qx^2-1) \pa^3 + \ell_{1,2} \pa^2 + \ell_{1,1} \pa + \ell_{1,0}, \\
  P^{\mathrm{pretz}}_{-2} &= (qx-1)(qx+1)(qx^2-1) \pa^3 + \ell_{2,2} \pa^2 + \ell_{2,1} \pa + \ell_{2,0}, \\
  P^{\mathrm{pretz}}_{2} &= q^{59} (qx-1) (q^2x+1) \pa^6 + \ell_{3,5} \pa^5 + \ell_{3,4} \pa^4 + \cdots + \ell_{3,0},
\end{align*}
where $\ell_{i,j} \in \bK[q][x]$. We now apply our desingularization algorithm
to each of the four operators.
\begin{enumerate}
\item By using Theorem~\ref{THM:orderbound}, we obtain an order bound~$b$ for
  a desingularized operator (see Table~\ref{TAB:exknots}).
\item Using Gr\"obner bases, we can find a generating set of~$M_b$.  Since the
  size of this generating set is large, we do not display it here.
\item By~\cite[Lemma 3.3.3]{Zhang2017}, we find the generator of~$I_b$ (see
  Table~\ref{TAB:exknots}).
\item It is straightforward to see that in each of the four cases, this single
  generator is the element in $I_b$ with minimal degree in~$x$.
\item Tracing back to the computation of steps (3) and (4), we find a
  $q$-difference operator $L \in \bK[q][x, \pa]$ of $\cont(P)$, which is of
  the following form:
  \begin{align*}
    L^{\mathrm{twist}}_{-1} & = q^4 x^2 \pa^3 - \left(q^9 x^4 - q^7 x^3 - q^5 x^3 - q^5 x^2 - q^4 x^2 - q^2 x + 1\right) \pa^2 -{} \\
    & \qquad q^4 x \left(q^4 x^4 - q^3 x^3 - q^3 x^2 - q^2 x^2 - q^2 x - x + q\right) \pa + q^7 x^3, \\
    L^{\mathrm{twist}}_{2} &= \pa^5 + p_{1,4} \pa^4 + p_{1,3} \pa^3 + \cdots + p_{1,0}, \\
    L^{\mathrm{pretz}}_{-2} &=  \pa^5 + p_{2,4} \pa^4 + p_{2,3} \pa^3 + \cdots + p_{2,0}, \\
    L^{\mathrm{pretz}}_{2} &=  \pa^{10} + p_{3,9} \pa^9 + p_{3,8} \pa^8 + \cdots + p_{3,0},
  \end{align*}
  where $p_{i,j} \in \bK[q][x]$.
\end{enumerate} 
We observe that in all four examples the minimal order for desingularized
operators matches with the predicted order bound, i.e., the bound is tight in
these cases. This can be seen by inspecting the $(b-1)$-st coefficient
ideal~$I_{b-1}$ (see Table~\ref{TAB:exknots}). We conclude that the sequences
that are annihilated by the four operators, respectively, consist indeed of
(Laurent) polynomials, provided that the initial values have this property as
well.
\end{ex}
\begin{table}
  \setlength{\tabcolsep}{10pt}
  \begin{center}
  \begin{tabular}{l|cccc}
    & $P^{\mathrm{twist}}_{-1}$ & $P^{\mathrm{twist}}_{2}$ & $P^{\mathrm{pretz}}_{-2}$ & $P^{\mathrm{pretz}}_{2}$ \\[0.5ex] \hline
    \rule{0pt}{12pt}order bound $b$ & 3 & 5 & 5 & 10 \\[1ex]
    generator of $I_b$ & $I_3=\langle x^2\rangle$ & $I_5=\langle 1\rangle$
      & $I_5=\langle 1\rangle$ & $I_{10}=\langle 1\rangle$ \\[1ex]
    generator of $I_{b-1}$ & & 
    $\langle q^3x^2-1 \rangle$ & 
    $\langle q^3x^2-1 \rangle$ & 
    $\langle q^4x-1 \rangle$
  \end{tabular}
  \end{center}
  \caption{Computations for Example~\ref{EX:knots}}
  \label{TAB:exknots}
\end{table}

\begin{ex} \label{EX:knots2}
By applying our desingularization algorithm to the unnormalized
$q$-recurrences of $J^{\mathrm{twist}}_{p}(n)$ for the same values of~$p$ as
in the previous example, we can prove that in these instances the operators
are not completely desingularizable, therefore confirming part (2) of
Conjecture~\ref{conj.Stavros}.
\end{ex}

Since Algorithms~\ref{ALGO:qWeylclosure} and~\ref{ALGO:desinop} involve
Gr\"obner bases computations, it is rather inefficient to find desingularized
operators when the size of the given $q$-difference operator is large.
Alternatively, we may apply guessing~\cite{Kauers2009a} to compute a
desingularized operator of a given $q$-difference operator, once we derive an
order bound by Theorem~\ref{THM:orderbound}.

In order to illustrate the guessing approach, we focus on a slightly modified
problem, namely that of finding \emph{bimonic recurrence equations}: we want
to completely desingularize both the leading and the trailing coefficient,
i.e., after desingularization these two coefficients should have the form
$q^{an+b}$ for some integers $a,b\in\bN$. The existence of such a recurrence
equation certifies that the bi-infinite sequence
$\bigl(f(n)\bigr){}_{n\in\bZ}$ has only Laurent polynomial entries. Note that
this approach is also suited for inhomogeneous recurrences.

It works as follows: assume we are given a (possibly inhomogeneous) recurrence
\[
  \underbrace{p_{-1}(q,q^n) + \sum_{i=0}^r p_i(q,q^n) f(n+i)}_{\textstyle =: R(n)} = 0,
\]
with $r\geq0$ and $p_i\in\bK[q,q^n]$ for $-1\leq i\leq r$. Define the polynomial
$c(q,q^n)\in\bK[q,q^n]$ by
\[
  c(q,q^n) = \frac{\lcm\bigl(p_0(q,q^n),p_r(q,q^n)\bigr)}{q^{an+b}}
\]
with integers $a,b\in\bN$ chosen such that $c(q,q^n)$ is neither divisible
by~$q$ nor by~$q^n$. The goal is to determine polynomials
$u_i(q,q^n)\in\bK[q,q^n]$ such that the coefficients $\ell_i(q,q^n)$, $-1\leq
i\leq r+s$, in the linear combination
\[
  \sum_{i=0}^s u_i(q,q^n) R(n+i) =
  \ell_{-1}(q,q^n) + \sum_{i=0}^{r+s} \ell_i(q,q^n)f(n+i)
\]
are all divisible by $c(q,q^n)$. Hence, we make an ansatz for the coefficients
of the linear combination, instead of trying to guess the desingularized
operator directly. The latter would be much more costly to compute (compare
the number of green dots with the number of blue dots in
Figure~\ref{FIG:support}).  The procedure is sketched in
Algorithm~\ref{ALGO:guess}. We have implemented it in Mathematica; the source
code and a demo notebook are freely available as part of the supplementary
electronic material~\cite{ElectronicCK}.
\newpage

\begin{algo} \label{ALGO:guess}
Given a recurrence $R(n)=p_{-1}(q,q^n) + \sum_{i=0}^r p_i(q,q^n) f(n+i)$ and a
factor $c(q,q^n)$ that is to be removed. Compute $u_i\in\bK[q,q^n]$ such that
$\sum_{i=0}^s u_i(q,q^n)R(n+i) =
c(q,q^n)\bigl(\ell_{-1}(q,q^n)+\sum_{i=0}^{r+s} \ell_i(q,q^n)f(n+i)\bigr)$ for
some polynomials $\ell_i\in\bK[q,q^n]$.
\begin{enumerate}
\item Make an ansatz of the form $A=\sum_{i=0}^s \sum_{j=e_i}^{d_i} c_{i,j}(q)
  q^{jn} R(n+i)$ (one may note that the coefficients $c_{0,j}$ and $c_{s,j}$
  are already prescribed (up to a constant multiple in $\bK(q)$) by the choice
  of $c(q,q^n)$.
\item Write $A$ in the form $A=a_{-1}(q,q^n)+\sum_{i=0}^{r+s} a_i(q,q^n) f(n+i)$.
\item For $-1\leq i\leq r+s$ compute the remainder of the polynomial division
  of $a_i(q,q^n)$ by $c(q,q^n)$, regarded as polynomials in~$q^n$.
\item Perform coefficient comparison in these remainders with respect to~$q^n$.
\item Solve the resulting linear system over $\bK(q)$ for the unknowns
  $c_{i,j}\in\bK[q]$ (we may clear denominators since the system is
  homogeneous).
\item Return $u_i(q,q^n)=\sum_{j=e_i}^{d_i} c_{i,j}(q) q^{jn}$.
\end{enumerate}
\end{algo}

\begin{figure}
  \begin{center}
    \includegraphics[width=0.9\textwidth]{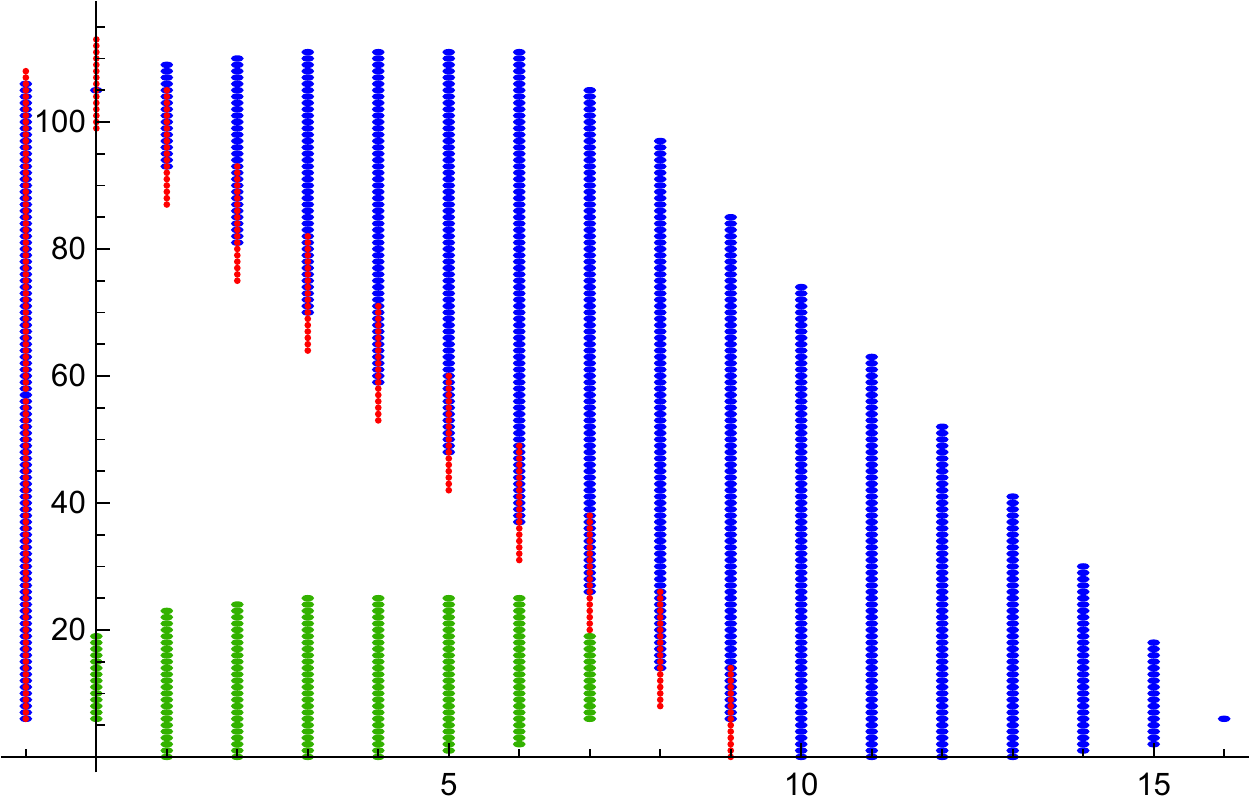}
  \end{center}
  \caption{$q^n$-support of the coefficients $p_{-1},\dots,p_9\in\bQ[q,q^n]$
    of the inhomogeneous $q$-recurrence for $K^{\mathrm{pretz}}_3$ (red), $q^n$-support of
    the coefficients $u_0,\dots,u_7$ (green), and $q^n$-support of the resulting bimonic
    recurrence (blue), represented by the coefficients $\ell_{-1},\dots,\ell_{16}$; the
    horizontal axis gives the index of the coefficient, the vertical axis the exponent
    of~$q^n$.}
  \label{FIG:support}
\end{figure}

It is interesting to note that our computed bimonic recurrences reveal certain
symmetries in their coefficients, more precisely, they are kind of palindromic.
For example, the bimonic $q$-recurrence that we found for $J^{\mathrm{pretz}}_{-2}(n)$,
written in the form
\[
  \sum_{j=0}^9 \ell_{-1,j}(q) q^{jn} + \sum_{i=0}^5 \sum_{j=0}^9 \ell_{i,j}(q) q^{jn} f(n+i)
\]
has the following palindromicity properties ($0\leq i\leq 5$, $0\leq j\leq 9$).
\[
  \ell_{i,j} = q^{5j-i-20} \ell_{5-i,9-j}
  \quad\text{and}\quad
  \ell_{-1,j} = -q^{5j-25} \ell_{-1,10-j}.
\]
This phenomenon is illustrated in Table~\ref{TAB:pretz-2}. It is also visible
in Figure~\ref{FIG:support} but on a different example. The occurrence of
palindromic operators in the context of knot theory has been studied in more
detail in~\cite{GaroufalidisKoutschan13}. Indeed, if we use the bimonic
recurrence to define the sequence $f(n)=(q^n-1)J^{\mathrm{pretz}}_{-2}(n)$ for
$n\leq0$ then we see that this sequence is palindromic:
\[
  f(n) = -q^n \, f(-n) \quad\text{for all }n\in\bN.
\]

\begin{table}
\[
\begin{array}{@{}l|r@{\,\cdot\,}lr@{\,\cdot\,}lr@{\,\cdot\,}lr@{\,\cdot\,}lr@{\,\cdot\,}lr@{\,\cdot\,}lr@{\,\cdot\,}l@{}}
  & \multicolumn{2}{c}{1} & \multicolumn{2}{c}{f(n)} & \multicolumn{2}{c}{f(n+1)} &
  \multicolumn{2}{c}{f(n+2)} & \multicolumn{2}{c}{f(n+3)} & \multicolumn{2}{c}{f(n+4)} &
  \multicolumn{2}{c}{f(n+5)} \\[0.5ex] \hline \rule{0pt}{12pt}
  q^{0n} & \multicolumn{2}{c}{} & \multicolumn{2}{c}{} & \multicolumn{2}{c}{} &
    \multicolumn{2}{c}{} & -1 & 2^4 & 1 & 2^3 & \multicolumn{2}{c}{} \\
  q^{1n} & \multicolumn{2}{c}{} & \multicolumn{2}{c}{} & \multicolumn{2}{c}{} &
    -1 & 2^{12} & 18 & 2^7 & -1 & 2^1 & \;1 & 2^0 \\
  q^{2n} & 1 & 2^9 & \multicolumn{2}{c}{} & \multicolumn{2}{c}{} & 33 & 2^{10} &
    13 & 2^7 & 24 & 2^5 & \multicolumn{2}{c}{} \\
  q^{3n} & 153 & 2^8 & \multicolumn{2}{c}{} & 1 & 2^{16} & 5 & 2^{12} & 177 & 2^9 &
    3 & 2^9 & \multicolumn{2}{c}{} \\
  q^{4n} & 93 & 2^{11} & \multicolumn{2}{c}{} & -1 & 2^{17} & 89 & 2^{13} & 3 &
    2^{14} & 3 & 2^{13} & \multicolumn{2}{c}{} \\
  q^{5n} & \multicolumn{2}{c}{} & \multicolumn{2}{c}{} & 3 & 2^{17} & 3 & 2^{17} &
    89 & 2^{15} & -1 & 2^{18} & \multicolumn{2}{c}{} \\
  q^{6n} & -93 & 2^{16} & \multicolumn{2}{c}{} & 3 & 2^{18} & 177 & 2^{17} & 5 &
    2^{19} & 1 & 2^{22} & \multicolumn{2}{c}{} \\
  q^{7n} & -153 & 2^{18} & \multicolumn{2}{c}{} & 24 & 2^{19} & 13 & 2^{20} & 33 &
    2^{22} & \multicolumn{2}{c}{} & \multicolumn{2}{c}{} \\
  q^{8n} & -1 & 2^{24} & 1 & 2^{20} & -1 & 2^{20} & 18 & 2^{25} & -1 & 2^{29} &
    \multicolumn{2}{c}{} & \multicolumn{2}{c}{} \\
  q^{9n} & \multicolumn{2}{c}{} & \multicolumn{2}{c}{} & 1 & 2^{27} & -1 & 2^{27} &
    \multicolumn{2}{c}{} & \multicolumn{2}{c}{} & \multicolumn{2}{c}{} \\
\end{array}
\]
\caption{Coefficients of the bimonic $q$-recurrence for
  $J^{\mathrm{pretz}}_{-2}(n)$; for space reasons only the evaluations for
  $q=2$ are given. In order to reveal the underlying symmetry, common powers
  of~$q$ are kept as powers of~$2$; for example, the entry $24\cdot2^5$ in the
  last-but-one column comes from the coefficient of $q^{2n}f(n+4)$ which is
  $q^5(q^4+q^3-q+2)$. The first column corresponds to the inhomogeneous part.}
\label{TAB:pretz-2}
\end{table}

We have applied Algorithm~\ref{ALGO:guess} to all recurrences associated to the
twist knots $K^{\mathrm{twist}}_p$ for $-14\leq p\leq15$ and to some of the pretzel
knots $K^{\mathrm{pretz}}_p$. All these results can be found in the supplementary
electronic material~\cite{ElectronicCK}.

\section{Conclusion} \label{SECT:conclusion} 
In this paper, we determine a generating set of the $q$-Weyl closure of a
given univariate $q$-difference operator, and compute a desingularized
operator whose leading coefficient has minimal degree in~$q$. Moreover, we use
our algorithms to certify that several instances of the colored Jones
polynomial are Laurent polynomial sequences.  A challenging topic for future
research would be to consider the corresponding problems in the multivariate
case.

Another direction of research we want to consider in the future is the
desingularization problem for linear Mahler equations~\cite{Mahler1929}, which
attracted quite some interest in the computer algebra community recently, see
for example~\cite{Chyzak2016}.  Mahler equations arise in the study of
automatic sequences, in the complexity analysis of divide-and-conquer
algorithms, and in some number-theoretic questions.

\paragraph{Acknowledgment}
The authors would like to thank Stavros Garoufalidis for providing the
examples of Section~\ref{SECT:application} and for enlightening discussions on
the knot theory part. We are also grateful to the anonymous referee for the
detailed report and for numerous valuable comments.

\bibliographystyle{abbrv}

\end{document}